\documentclass[conference]{IEEEtran}
\IEEEoverridecommandlockouts
\usepackage{cite}
\usepackage{amsmath,amssymb,amsfonts}
\usepackage{algorithm}
\usepackage[noend]{algpseudocode}
\usepackage{graphicx}
\usepackage{textcomp}
\usepackage{mathrsfs}
\usepackage{amsthm}
\usepackage{mathtools}
\theoremstyle{plain}
\newtheorem{theorem}{Theorem}[section]

\newtheorem{proposition}[theorem]{Proposition}
\theoremstyle{definition}

\theoremstyle{remark}
\DeclareMathOperator*{\argmax}{arg\,max}

\newcommand{\sign}{\text{sgn}}
\usepackage{xcolor}
\def\BibTeX{{\rm B\kern-.05em{\sc i\kern-.025em b}\kern-.08em
    T\kern-.1667em\lower.7ex\hbox{E}\kern-.125emX}}
\usepackage[]{geometry}
\allowdisplaybreaks
\begin{document}
\title{\vspace{6.65mm} Motion Planning for Identification of Linear Classifiers}
\author{Aneesh Raghavan and Karl Henrik Johansson
\thanks{*Research supported by the Swedish Research Council (VR), Swedish Foundation for Strategic Research (SSF),  and the Knut and Alice Wallenberg Foundation. The authors are with the Division of Decision and Control Systems, 
Royal Institute of Technology,  KTH,  Stockholm. 
Email: {\tt\small aneesh@kth.se,  kallej@kth.se}}%
}
\maketitle
\begin{abstract}
A given region in 2-D Euclidean space is divided by a unknown linear classifier in to two sets each carrying a label. The objective of an agent with known dynamics traversing the region is to identify the true classifier while paying a control cost across its trajectory. We consider two scenarios: (i) the agent is able to measure the true label perfectly; (ii) the observed label is the true label multiplied by noise. We present the following: (i)  the classifier identification problem formulated as a control problem; (ii) geometric interpretation of the control problem resulting in one step modified control problems; (iii) control algorithms that result in data sets which are used to identify the true classifier with accuracy; (iv) convergence of estimated classifier to the true classifier when the observed label is not corrupted by noise; (iv) numerical example demonstrating the utility of the control algorithms.  
\end{abstract}
\section{Introduction}\label{section 1}
\subsection{Motivation}\label{subsection 1.1}
Duality between control and learning (in a broad sense, including estimation and inference problems) has been well studied in the literature. In \cite{todorov2008general}, duality between estimation and control is studied for general stochastic control problems. In \cite{ishii2002control}, exploration vs exploitation has been studied through the control of a meta- parameter in reinforcement learning problems. \cite{alpcan2015information} studies dual control problems where knowledge gained through control actions is  explicitly defined. In \cite{klenske2016dual}, dual control techniques have been applied to approximate the intractable aspects of Bayesian RL, leading to structured exploration strategies that differ from standard RL. In \cite{mesbah2018stochastic}, stochastic model predictive control is presented in the dual control paradigm. More recently dual control has been applied to active uncertainty learning in human robot interaction, \cite{hu2022active}. In all these problems there, is uncertainty in the model or the cost function that is being actively learnt through control actions. 

Parallels between model predictive control and algorithms in A.I have been drawn, \cite{lecun2022path}. Learning theory has been extensively applied to control problems; special neural networks and deep networks have been used extensively in system identification and  to approximate solutions to control problems, \cite{becerikli2003intelligent}, \cite{sabouri2017neural}, \cite{sanchez2018real}. Systems and control theory however has not been applied to its full potential to learning theory. A hypothesis that is being explored currently is, learning problems could benefit from being formulated as control problems by channelizing feedback to reduce the quantity of data required to solve the learning problem efficiently and accurately.

Adaptive sampling is closely related to the field of active learning, however the former operates in the context supervised learning while the latter is associated with semi-supervised learning. Adaptive sampling for classification has been studied in \cite{djouzi2022new}, \cite{singh2017sequential}, \cite{shekhar2021adaptive},  where sequential sampling algorithms are presented to enhance the learning process.  In \cite{ding2021adaptive}, adaptive sampling has been applied to hyperspectral image classification leading to improvement from state of the art. Learning unknown environment is a crucial part of marine robotics.  Adaptive sampling methods have been used to survey and learn about algal bloom,  water quality models, etc in \cite{zhang2007adaptive}, \cite{bernstein2013learning}, \cite{stankiewicz2021adaptive}, and \cite{fossum2020compact}.  

Given the above context, the problem that we consider is the identification of certain aspects of an agents environment that is unkown. Unlike traditional dual control which deals with uncertainty in the system or the cost functions and learning the same, we consider learning of the environment. The problem considered  has potential application in marine robotics as well, as described above. In our previous work, \cite{raghavan2023motion} we considered path planning for identification of functions in an agents surroundings. 
\subsection{Problem Considered}\label{subsection 1.2}
The problem considered is as follows. A given region in two dimensional space is divided into two regions by a straight line, Figure \ref{Figure 1}. The true classifier divides the state space into two sets, $\mathcal{X}$, where the true label is $1$ and, $\mathcal{X}^c$, where the true label is $-1$. Every point in the region ``orange" carries the label $1$ while every point in the alternate region ``blue" carries the label $-1$. Every point on the straight line that divides the region carries the label $0$. Four points, $p_{1}, p_{2}, p_{3}$, and $p_{4}$, with their true labels are given  by an oracle. In Figure \ref{Figure 1}, the true labels of $p_{1}$ and $p_{4}$ are $1$ while that of $p_{2}$ and $p_{3}$ are $-1$. The true classifier is parameterized by its slope, $\rho^*$, and intercept, $c^*$. An agent with known dynamics traverses the region by paying a control cost. We consider two measurement models: (i) the agent is able to measure the true label perfectly (deterministic) (ii) the measurement gets corrupted by noise; the measured label is the true label flipped ($1$ to $-1$ and $-1$ to $1$) with a certain probability (stochastic). Given the measurement model, the objective of the agent starting at the point $p_{1}$ is to follow a trajectory during which it collects $m$ data points  which are optimal for the identification of the classifier while simultaneously minimizing its control costs.

One possible path that could be taken by the agent is depicted in Figure \ref{Figure 1}. The agent collects $6$ data points, apart from the $4$ given, $3$ of each label. The $10$ data points would subsequently be used to estimate the classifier. The four initial points provided by the oracle are assumed to be ``far" apart from each other. They provide an initial estimate of the classifier and a region  of the $2-D$ space to be explored by the agent to refine the estimate and identify the  classifier more accurately. The key idea that we would like to explore is that rather than collecting large number of data samples from both regions for accurate estimation, is it possible to strategically sample few data points which leads to the same accurate estimation as the large data set. Thus, problem considered here can be interpreted as an adaptive sampling problem which is restricted by the agent's dynamics.    
\begin{figure}
\begin{center}
\includegraphics[width=\columnwidth]{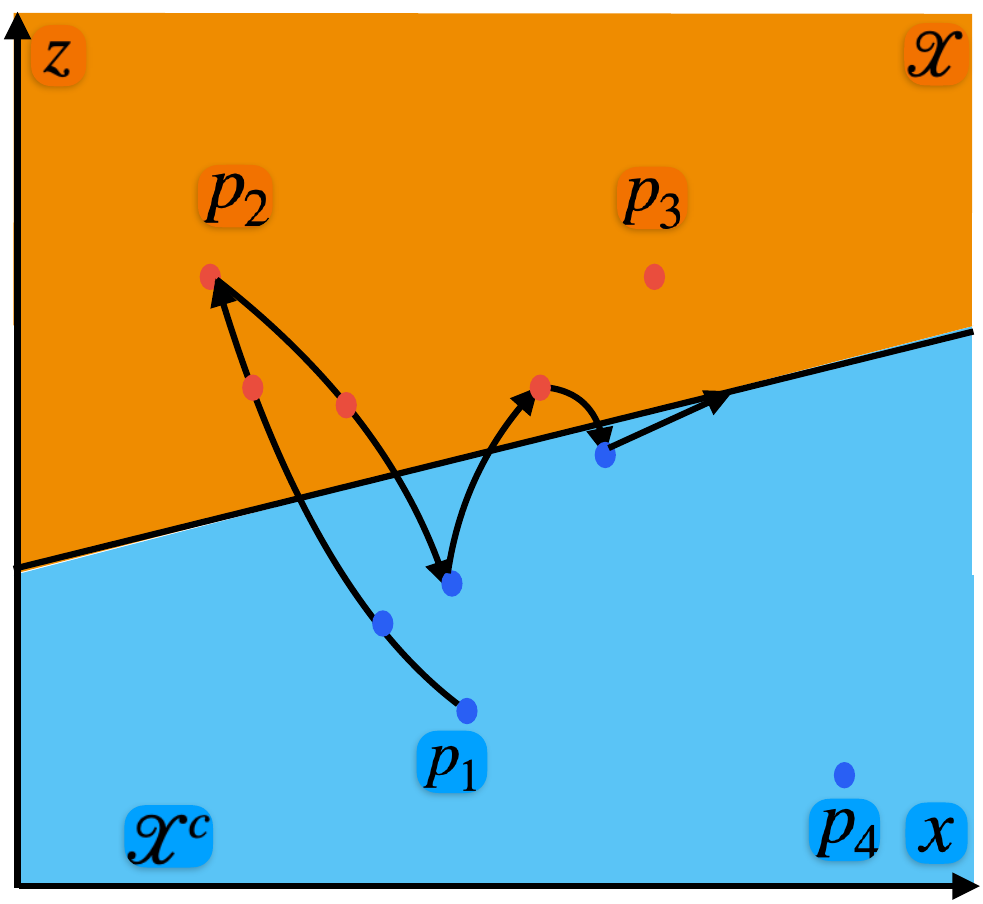}
\caption{Schematic for the motion planning problem} 
\label{Figure 1}
\end{center}
\vspace{-1cm}
\end{figure}
\subsection{Contributions}\label{subsection 1.3}
We formulate the identification problem of the true classifier as described above has a control problem in both the deterministic and stochastic scenarios. For the latter, we present an explicit construction of the probability space. The formulated control problems are analyzed and the associated challenges and drawbacks are presented. We present a geometric interpretation of the problem using 2D analytic geometry. Utilizing the geometric ideas, we formulate one step control problems which can be solved in a numerically efficient manner. For the stochastic scenario, we derive the equations for propagation of the probability that the true parameters characterizing the true classifier belongs to a given set in the parameter space. We present separate control algorithms (a ``greedy" approach) for noiseless and noisy data to solve the identification problem which involves solving the one step control problems formulated before. The data set obtained by executing the control algorithm is used to identify the classifier. In the deterministic case, we prove that the classifier identified converges to the true classifier, i.e., the estimated parameters of the classifier converges to the true classifier. In stochastic case, the data collected results in sets to which the true parameters belongs to with high probability. We present an example for both cases illustrating the control algorithms and the resulting classifiers. 
\subsection{Outline}\label{subsection 1.4}
In Section \ref{section 2}, we present the formulation of the identification problem as a control problem. In Section \ref{section 3}, we present the geometric interpretation of the identification problem and the associated one step control problems. In section \ref{section 4}, we present the control algorithms and the proof of convergence to true classifier in the deterministic scenario. In Section \ref{section 5}, we present a numerical example for both scenarios which demonstrates the application of the control algorithm. In Section \ref{section 6}, we summarize the work presented in the paper and discuss future work.
Given a collection of sets $\mathfrak{B}$, we denote the smallest $\sigma$ algebra generated by it as $\sigma(\mathfrak{B})$.
\section{Problem Formulation}\label{section 2}
In this section, we present the abstraction of the problem discussed in subsection \ref{subsection 1.2} . 
\subsection{Deterministic Scenario: Abstraction}\label{subsection 2.1}
\subsubsection{Identification of the Classifier}\label{subsubsection 2.1.1}
By executing a control policy and following the corresponding path, the positions at which measurements are collected by the agent is denoted by $\{x_{j}, z_{j}\}^{m}_{j=1}$. The true labels of these positions is denoted by $\{y_{j}\}^m_{j=1}$. If the true classifier is denoted by the straight line $z = \rho^* x + c^*$, then $y_j = \sign(z_j - \rho^*x_j -c^*), j=1, \ldots, m$. Given the data points, the classification problem can be formulated as, 
\begin{align*}
\underset{\rho \in \mathbb{R}, c\in \mathbb{R}} \min \sum^{m}_{j=1} y_{j}\sign(z_j - \rho x_j -c).
\end{align*}
The optimization problem is to choose parameters $\rho$ and $c$ so that the resulting line classifies as many data points as possible correctly. This formulation does not yield to computationally effective solutions. The classical reformulation, \cite{paulsen2016introduction}, \cite{hofmann2008kernel}, of the problem is to identify a line which (i) maximizes the minimum distance of the data points from the line; (ii) ensures that the points are classified correctly. Given a possible classifier, $z - \rho x -c = 0$, the distance of the point $(x_k,z_{k})$ is given by $\frac{|z_{k} - \rho x_{k} -c |}{\rho^2}$. Through suitable re-normalization, it can be ensured that  $|z_{j^*} - \rho x_{j^*} - c| = 1$ for atleast one of the data points,$(x_{j^*}, z_{j^*}, y_{j^*})$, while  $|z_{j} - \rho x_{j} - c| \geq  1$ for all other data points. The classification problem can be reformulated as, 
\begin{align*}
\underset{\rho \in \mathbb{R}, c\in \mathbb{R}} \min\;\; \frac{\rho^2}{2} \;\; \text{s.t} \;\;  y_{j}(z_{j} - \rho x_{j} - c )\geq  1, j=1, \ldots, m.
\end{align*}
The above optimization problem is a convex optimization problem which can be solved effectively by considering the dual problem, \cite{hofmann2008kernel}. 
\subsubsection{Control Problem}\label{subsubsection 2.1.2}
The state space of the agent is $\mathbb{R}^2$, while the set of actions (action space) that it can take is denoted by $\mathcal{U}$. The state of the system at time $t$ is denoted by $\zeta(t) = [x(t), z(t)]$ while the action is denoted by $u(t)$. Along with the full state information at $t$, the observation of the system at $t$ is the true label at that state which is denoted by $y(t)$. Thus, $y(t)= \sign (z(t) - \rho^{*} x(t) -c^{*})$ where $\rho^*$ and $c^*$ the parameters corresponding to the true classifier and are unknown. The observation $y(t)$ is available to the agent and is used to estimate $\rho^*$ and $c^*$. The initial state of agent is the point $p_{1}$ which we denote as $(\bar{x}_1, \bar{z}_1, -1)$. For the deterministic case,  we do not consider process noise or measurement noise. The agent is modeled as a discrete time system with known dynamics, 
\begin{align*}
&\zeta(t+1) =\phi(\zeta(t), u(t)),\; \bar{y}(t) = [\zeta(t), y(t)],\;  \bar{y}(0) = [\bar{x}_1, \bar{z}_1,   \\
& -1], \; y(t)= \sign (z(t) - \rho^{*} x(t) -c^{*}), t = 0, \ldots m-1. 
\end{align*}
The learning problem formulated as a control problem is to find a control policy, $\{\Upsilon_{j}\}^{m-1}_{j=0}$, where $\Upsilon_j: \mathbb{R}^{2^{j+1}} \times\{-1,0,1\}^{j+1} \to \mathcal{U}$, which minimizes the control cost while ensuring that the data points collected yield a suitable classifier. The optimization problem is, 
\begin{align*}
&\underset{\{\Upsilon_{t}\}^{m-1}_{t=0}, \{\rho,c\} \subset \mathbb{R}} \min  \frac{\rho^2}{2} + \sum^{m-1}_{t=0} \Big|\Big| \Upsilon_{t}\Big( \{\bar{y}(j)\}^{t-1}_{j=0}\Big) \Big|\Big|^2, \; \text{s.t}, \; \\
&\  y(t)(z(t) - \rho x(t) - c )\geq  1, t=1, \ldots, m, \; |\sum^{m}_{t=0} y(t)| \leq 1, \\
&(\bar{z}_j - \rho \bar{x}_j - c ) \leq  1, j=1,4, (\bar{z}_j - \rho \bar{x}_j - c )\geq  1,j=2,3, \\
&\zeta(t+1) = \phi \Big(\zeta(t), \Upsilon_{t}\Big( \{\bar{y}(j)\}^{t-1}_{j=0}\Big) \Big), t = 0, \ldots m-1, y(t)\\
&= \sign (z(t) - \rho^{*} x(t) -c^{*}), \bar{y}(t) = [\zeta(t), y(t)], t = 0, \ldots m.
\end{align*}
In the constraints of the above optimization problem, in the second line, the constraints are to ensure that the four given points are classified correctly. The constraint, $|\sum^{m}_{t=0} y(t)| \leq 1$, has been included to ensure that equal number points are visited in both regions or at most one additional point is visited in one of the regions.  
\subsection{Stochastic Scenario: Abstraction}\label{subsection 2.2}
\subsubsection{The Noise Model}\label{subsubsection 2.2.1}
First, we describe the noise model. In the region close to the true classifier, the measured labels are equal to true labels with probability $p$ (noise = $1$) and get flipped to the other label with probability $1-p$ (noise = $-1$). In the region far away from the true classifier, the measured label equals the true label with probability $1$. Given the parameters of the true classifier, $\rho^*$ and $c^*$, and some $\bar{c} >0$, large enough, define the set $\hat{E}$, as $\hat{E} =\{(x,z) \in \mathbb{R}^2 : z- \rho^* x -c^*+\bar{c} \geq 0 \; \text{and}  \; z- \rho^* x -c^*-\bar{c} \leq 0\}$. The region has been depicted in Figure \ref{Figure 2}.  Then, one possible model for noise is 
\begin{align*}
\varepsilon(x,z) = 
\begin{cases}
 1 & \text{with probability $p$ if $(x,z) \in \hat{E}$}\\
-1 & \text{with probability $1-p$ if $(x,z) \in \hat{E}$}\\
1  & \text{with probability 1 if $(x,z) \in \mathbb{R}^2 \sim \hat{E}$} 
\end{cases}.
\end{align*}
\begin{figure}
\begin{center}
\includegraphics[width=\columnwidth]{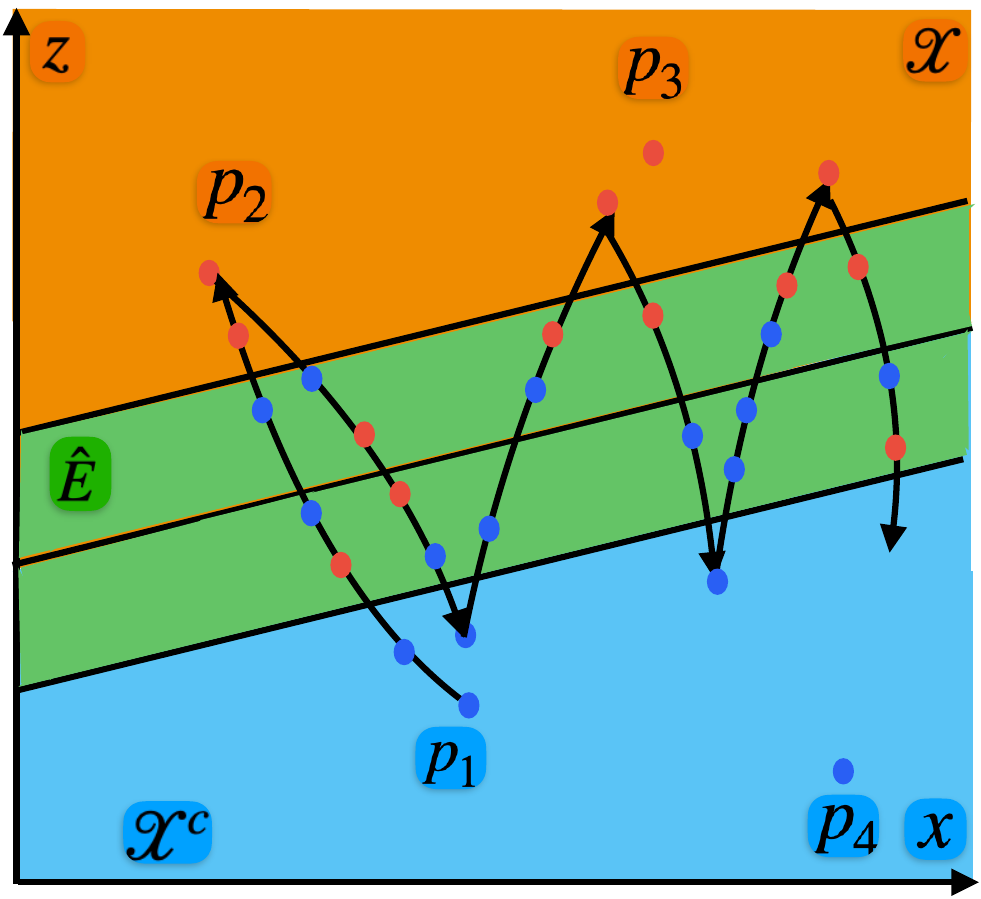}
\caption{Schematic for the motion planning problem with noisy data} 
\label{Figure 2}
\end{center}
\vspace{-0.9cm}
\end{figure}
The region $\hat{E} \cap \mathcal{X}$ carries the true label $1$ (indicated in red). However, due to noise some of the observed labels are $-1$. There are four data points whose true label is $1$ while the observed label is $-1$. In the region, $\hat{E} \cap \mathcal{X}^c$ the true label is $-1$. There are $3$ data points whose true label is $-1$ while the observed label is $1$. Note that this model for noise requires the knowledge about the true classifier which is unknown and hence cannot be used. Since it is difficult to define the regions ``close to" and ``far off" with out invoking the definition of the true classifier, we consider an alternate noise model. We assume that we are given a sequence (or a large set) of points, $\{p_{j}\}_{j \geq 1} (\{p_{j}\}^{I}_{j = 1})$, such that a subsequence (r a subset) of which lies in the region with label $1$ while the complement subsequence (subset) lies in the region with label $-1$. Around each point, $p_{j}=(x_j, z_j)$, there exists a region , $B_{r}(p_{j}) = \{(x,z): || (x,z) - (x_j,z_j) || \leq  r \}$, where the measured label equals the true label with probability 1. For the remaining state space, the noise takes values $1$ and $-1$ with probability $p$ and $1-p$ respectively. Stating the same precisely, let $\bar{E} =\underset{j \geq 1}\bigcup B_{r}(p_{j})$, then, 
\begin{align*}
\varepsilon(x,z) = 
\begin{cases}
 1 & \text{with probability 1 if $(x,z) \in \bar{E}$} \\
 1 & \text{with probaility $p$ if $(x,z) \in \mathbb{R}^2 \sim \bar{E}$}\\
-1 & \text{with probaility $1-p$  if $(x,z) \in \mathbb{R}^2 \sim \bar{E}$}
\end{cases}.
\end{align*}
\begin{figure*}
\hrulefill 
\begin{align}
&\bar{\mathbb{P}}_{0}((X(0),Z(0)) \in E, Y(0) =1) = \bar{\mathbb{P}}_{0}((X(0),Z(0)) \in E \cap \mathcal{X}, Y(0) =1) = \bar{\mathbb{P}}_{0}(X(0),Z(0) \in E \cap \mathcal{X}), \nonumber \\ 
&\bar{\mathbb{P}}_{0}((X(0),Z(0)) \in E, Y(0) =-1) = \bar{\mathbb{P}}_{0}((X(0),Z(0)) \in E \cap \mathcal{X}^c, Y(0) =-1) = \bar{\mathbb{P}}_{0}(X(0),Z(0) \in E \cap \mathcal{X}^c). \label{Equation 1} \\
& \mathbb{P}_{0}((\zeta(0), Y(0)) \in E, U(0) \in F) \hspace{-3pt}=  \hspace{-3pt}\int_{E}\int_{F} \hspace{-3pt}\bigg[d\Gamma_{0}(x_{0}, z_0,y_0)[u] \bigg] d \bar{\mathbb{P}}_{0}(x_{0}, z_0,y_0) = \Gamma_{0}( E )[U(0) \in F]\bar{\mathbb{P}}_{0}( (\zeta(0), Y(0) \in E ) \label{Equation 2} \\
&\hat{\mathbb{P}}_{n-1}(\zeta(n) \in E) = \mathbb{P}_{n-1} \Big(\Big(\zeta(0), Y(0), \{U(j)\}^{n-1}_{j=0}, \{\varepsilon(j)\}^{n-1}_{j=1}\Big): \phi^{n}(\zeta(0), \{U(j)\}^{n-1}_{j=0}) \in E\}, E \in \mathcal{B}(\mathbb{R}^{2}) \label{Equation 3} \\
& \bar{\mathbb{P}}_{n}\Big(\Big(X(0),Z(0), Y(0), \{U(j)\}^{n-1}_{j=0}, \{\varepsilon(j)\}^{n-1}_{j=1}\Big) \in E , \varepsilon(n) = 1 \Big) \hspace{-3pt}= p \hat{\mathbb{P}}_{n-1}(\zeta(t) \in \phi^{n}(E) \cap  \mathbb{R}^2 \sim \bar{E}) + \hat{\mathbb{P}}_{n-1}(\zeta(t) \in \phi^{n}(E)  \nonumber \\
&\cap \bar{E}), \;  \bar{\mathbb{P}}_{n}\Big(\Big(X(0),Z(0), Y(0), \{U(j)\}^{n-1}_{j=0}, \{\varepsilon(j)\}^{n-1}_{j=1}\Big) \in E , \varepsilon(n) = -1 \Big) = (1 - p) \hat{\mathbb{P}}_{n-1}(\zeta(t) \in \phi^{n}(E) \cap  \mathbb{R}^2 \sim \bar{E})    \label{Equation 4} 
\end{align}
\hrulefill 
\vspace{-0.4cm}
\end{figure*}
\subsubsection{The Probability Space}\label{subsubsection 2.2.2}
Given the above noise model, we construct the probability space for the stochastic control problem as follows.  
\textit{Sample Space and $\sigma$ Algebra}: Let $\Omega = \mathbb{R}^2 \times \{-1,1\} \times \mathcal{U}^{\mathbb{Z}_{+}} \times \{-1,1\}^{\mathbb{N}}$. The sample space is the product set of the set of all possible values that can be taken by the initial states $X(0, \cdot), Z(0, \cdot)$, the initial observation $Y(0,\cdot)$, the control trajectory $\{U(t, \cdot)\}_{t \geq 0}$ and the sequence of noise variables $\{\varepsilon(t, \cdot)\}_{t \geq 1}$. By considering the Borel $\sigma$ algebra on $\mathcal{U}$, $\mathcal{B}(\mathcal{U})$, the Borel $\sigma$ algebra on $\mathbb{R}^2$, $\mathcal{B}(\mathbb{R}^{2})$, and the algebra $\mathcal{B}(\{-1,1\}) = \{\{1\}, \{-1\}, \{-1,1\}, \emptyset\}$ on $\{-1,1\}$, we define
\begin{align*}
&\mathfrak{B}_{0} =\{ F \subset  \mathbb{R}^2  \times \{-1,1\} \times \mathcal{U}  : \; F = E_{1} \times E_{2} \times E_3, \\
&\hspace{1.3cm}\text{where} \;E_{1} \in \mathcal{B}(\mathbb{R}^{2}),  E_{2} \in \mathcal{B}(\{-1,1\}), E_3 \in  \mathcal{B}(\mathcal{U}) \} \\
&\mathfrak{B} = \{F \subset \mathcal{U}  \times \{-1,1\}: F = E_{1} \times E_{2}, \; \text{where} \\
&\hspace{4.15cm}E_{1} \in \mathcal{B}(\mathcal{U}), E_{2} \in \mathcal{B}(\{-1,1\})  \}.
\end{align*}
We define the collection of cylindrical subsets of $\Omega$ as, 
\begin{align*}
\bar{\mathfrak{B}} = &\{ \bar{\omega} \in \Omega: \bar{\omega}(0) \in E_{0}, \bar{\omega}(1) \in E_{1}, \ldots, \bar{\omega}(n) \in E_{n}, \\
\; &\text{where}\; E_{0} \in \sigma(\mathfrak{B}_{0}),\{E_{j}\}^n_{j = 1} \in \sigma(\mathfrak{B}), n \in \mathbb{Z}_{+}\}
\end{align*}
Let $\mathcal{F}$ be the smallest $\sigma$ algebra generated by $\bar{\mathfrak{B}}$. We note that $\bar{\omega}(t)$ is a 4-tuple at $t=0$ which we denote as $\bar{\omega}(0) = [x_0,z_0,y_0,u_0]$ and is a 2-tuple for $t \geq 1$ which we denote as $\bar{\omega}(t) = [u_t,e_t]$. For any trajectory $\omega \in \Omega$, $U(t,\omega) = u_{t}$ and $\varepsilon(t, \omega) = e_{t}$, i.e., the components corresponding to the control and noise. 

Hence, given a measurable mapping $\phi(\cdot)$ corresponding to the dynamics of the agent, the sequence  $\{\zeta(t, \cdot)\}_{t\geq 0}$, generated as $\zeta(t+1, \omega) = \phi(\zeta(t,\omega), U(t,\omega))$ is a stochastic process on $(\Omega, \mathcal{F})$. We denote the two components of $\zeta(t,\cdot) $ as $\zeta(t,\cdot) = [X(t,\cdot), Z(t, \cdot)]$. The sequence of true labels at the states $\{\zeta(t, \cdot)\}_{t\geq 0}$, $\{Y(t, \cdot)\}_{t\geq 0}$ generated as $Y(t,\omega) = \sign(Z(t,\omega) - \rho^* X(t, \omega)) -c^*$ is also a stochastic process on $(\Omega, \mathcal{F})$, where $\rho^*$ and $c^*$ are the parameters of the true classifier. The sequence of observed labels, $\{\bar{Y}(t, \cdot)\}_{t\geq 0}$ generated as $\bar{Y}(t, \omega) = Y(t, \omega) \varepsilon(t,\omega)$ is a stochastic processes on $(\Omega, \mathcal{F})$.

\textit{The Probability Measure:} We construct the  measure on $(\Omega, \mathcal{F})$ as follows. The distribution of the initial state, $\zeta(0)$, $\bar{\mathbb{P}}_{0}(X(0),Z(0))$ is assumed to be known and the joint distribution with the initial observation is defined in Equations \ref{Equation 1}. The measure $\bar{\mathbb{P}}_{0}$ is then extended to the $\sigma$ algebra, $\sigma(\bar{\mathfrak{B}}_{0})$, where $\bar{\mathfrak{B}}_{0} = \{ F \subset  \mathbb{R}^2  \times \{-1,1\} : F = E_{1} \times E_{2} \; \text{where} \; E_{1} \in \mathcal{B}(\mathbb{R}^{2}),  E_{2} \in \mathcal{B}(\{-1,1\}) \}$.  The equations essentially state that $\bar{\mathbb{P}}_{0} (Y(0) =1 | (X(0),Z(0)) \in E \cap \mathcal{X}) =1$ and $\bar{\mathbb{P}}_{0} (Y(0) =-1| (X(0),Z(0)) \in E \cap \mathcal{X}^c) =1$ for $E \in \mathcal{B}(\mathbb{R}^{2})$. This states that the label at the initial state is known perfectly. Hence,  it is reasonable to assume that the distribution of the initial states is concentrated in a region far away from the true classifier or in the set $\bar{E}$ (defined in the noise model) where the true label is known with probability 1. 

The control policy at $t=0$, $\Gamma_{0} : \sigma(\bar{\mathfrak{B}}_{0})  \to \mathcal{P}(\mathcal{U}, \mathcal{B}(\mathcal{U}))$ is to be found, where $\mathcal{P}(\mathcal{U}, \mathcal{B}(\mathcal{U}))$ is the set of probability measures on the Borel $\sigma$ algebra of $\mathcal{U}$. Given the control policy at $t=0$, the measure at stage $t=0$ is defined for $E \in \sigma(\bar{\mathfrak{B}}_{0}), F \in \mathcal{B}(\mathcal{U})$ in Equation \ref{Equation 2}. The measure is then extended to the $\sigma$ algebra, $\sigma(\{F   \subset  \mathbb{R}^2  \times \{-1,1\} \times \mathcal{U} : F= E_{1} \times E_{2}, E_{1} \in \sigma(\bar{\mathfrak{B}}_{0}), E_{2} \in  \mathcal{B}(\mathcal{U}) \}) =  \sigma(\mathfrak{B}_{0}) \overset{\Delta}{=} \mathcal{F}_{0}$.

For any $n \in \mathbb{N}$, the distribution of state $\zeta(n)$ is defined in Equation \ref{Equation 3}. We note that given the sets to which $\zeta(0)$ and $\{U_{j}\}^{n-1}_{j=0}$ belong, we can precisely state the set to which $\zeta(n)$ belongs. However, to find the probability that $\zeta(n)$ belongs to a measurable set, the probability of events like $U(j) \in E$ is needed. We  note that probability  of $U(j) \in E$ depends on $\Big(\zeta(0), Y(0), \{U_{k}\}^{j-1}_{k=0}, \{\varepsilon(k)\}^{j}_{k=1}\Big)$, i.e, not only on the initial state and the past control values but also on the initial observation and the noise random variables. Hence, in Equation \ref{Equation 3} the random variables, $Y(0), \{\varepsilon(j)\}^{n-1}_{j=1}$, have also been included while finding the distribution of $\zeta(n)$. The distribution of the observation $Y(n)$ is given by, $\hat{\mathbb{P}}_{t-1}(Y(t) = 1) = \hat{\mathbb{P}}_{t-1}(\zeta(t) \in \mathcal{X}), \; \hat{\mathbb{P}}_{t-1}(Y(t) = -1) = \hat{\mathbb{P}}_{t-1}(\zeta(t) \in \mathcal{X}^c)$. Similarly, the joint distribution of $\zeta(t)$ and $Y(t)$ can be found. 
\begin{figure*}
\hrulefill 
\begin{align}
&\mathbb{P}_{n}\Big(\Big(\zeta(0), Y(0), \{U(j)\}^{n-1}_{j=0}, \{\varepsilon(j)\}^{n}_{j=1}\Big) \in E , U(n) \in F \Big) = \nonumber \\
&\hspace{4cm}\Gamma_{n}(E)[U(n) \in F ]\bar{\mathbb{P}}_{n}\Big(\Big(\zeta(0), Y(0), \{U(j)\}^{n-1}_{j=0}, \{\varepsilon(j)\}^{n}_{j=1}\Big) \in E \Big), E \in \mathcal{F}_{n-1}, F \in \mathcal{B}(\mathcal{U})   \label{Equation 5} \\
&\mathbb{P}_{n}\Big(\Big(\zeta(0), Y(0), \{U(j)\}^{n-1}_{j=0}, \{\varepsilon(j)\}^{n-1}_{j=1}\Big) \in E ,  \varepsilon(n) \in \{-1,1\},  U(n) \in \mathcal{U} \Big)  = \Gamma_n(E \times (\{-1,1\})[ U(n) \in  \mathcal{U}] \times  \nonumber \bar{\mathbb{P}}_{n}\Big(\Big(\zeta(0),  \\
&Y(0), \{U(j)\}^{n-1}_{j=0}, \{\varepsilon(j)\}^{n-1}_{j=1}\Big) \in E ,  \varepsilon(n) \in \{-1,1\}\Big) = p\hat{\mathbb{P}}_{n-1}(\zeta(t) \in \phi^{n}(E) \cap  \mathbb{R}^2 \sim \bar{E})+ \hat{\mathbb{P}}_{n-1}(\zeta(t) \in \phi^{n}(E) \cap \bar{E}) +   \nonumber  \\
& (1 - p) \hat{\mathbb{P}}_{n-1}(\zeta(t) \in \phi^{n}(E) \cap  \mathbb{R}^2 \sim \bar{E})  = \hat{\mathbb{P}}_{n-1}(\zeta(t) \in \phi^{n}(E) ) = \mathbb{P}_{n-1}\Big(\Big(\zeta(0), Y(0), \{U(j)\}^{n-1}_{j=0}, \{\varepsilon(j)\}^{n-1}_{j=1}\Big) \in E\Big)\label{Equation 6} \\
&\frac{p \hat{\mathbb{P}}_{n-1}(\zeta(t) \in \phi^{n}(\tilde{E}_{1}) \cap \bar{E}^c)  \hspace{-2pt} +  \hspace{-2pt} \hat{\mathbb{P}}_{n-1}(\zeta(t) \in \phi^{n}(\tilde{E}_{1} ) \cap \bar{E}) }{\hat{\mathbb{P}}_{n-1}(\zeta(t) \in \phi^{n}(\tilde{E}_{1} ))}   \hspace{-3pt} \neq \hspace{-3pt}   
\frac{p \hat{\mathbb{P}}_{n-1}(\zeta(t) \in \phi^{n}(E_1) \cap \bar{E}^c)  \hspace{-2pt} +  \hspace{-2pt} \hat{\mathbb{P}}_{n-1}(\zeta(t) \in \phi^{n}(E_1) \cap \bar{E}) }{\hat{\mathbb{P}}_{n-1}(\zeta(t) \in \phi^{n}(E_1))},\label{Equation 7}
\end{align}
\hrulefill 
\vspace{-0.4cm}
\end{figure*}

Given $\mathbb{P}_{n-1}$ on $(\Omega_{n-1}, \mathcal{F}_{n-1})$, $\mathbb{P}_{n}$ is inductively defined as follows. $\bar{\mathbb{P}}_{n-1}$ is defined for $E \in \mathcal{F}_{n-1}$ and $F \in \mathcal{B}(\{-1,1\})$ in Equation \ref{Equation 4}. It is then extended to $\sigma(\mathcal{F}_{n-1} \times \mathcal{B}(\{-1,1\}))$. Because of the joint distribution when $\varepsilon(n)=1$ in Equation 4, following conditional independence does not hold:
\begin{align*}
\mathbb{P}_{n}(\varepsilon(n) \in F | \zeta(n) \in E_{1}, \{\varepsilon(j)&\}^{n-1}_{j=1} \in E_{2}) \\
&\neq  \mathbb{P}_{n}(\varepsilon(n) \in F | \zeta(n) \in E_{1}),
\end{align*}
$F \in \mathcal{B}(\{-1,1\})$, $E_{1} \in \mathcal{B}(\mathbb{R}^2)$, and $E_{2} \in \mathcal{B}(\{-1,1\}^{n-1})$. This is because event $E_{2}$ interferes with event $E_1$ causing a change in the set of states which we denote by $\tilde{E}_1$,
\begin{align*}
\Big(\zeta(0), Y(0), \{U(j)\}^{n-1}_{j=0}, \{\varepsilon(j)\}^{n-1}_{j=1}\Big) \in E_{1} \Big) \bigcap \Big(\{\varepsilon(j)\}^{n-1}_{j=1}  \\
\in E_{2})  \overset{\Delta}{=}  \Big(X(0),Z(0), Y(0), \{U(j)\}^{n-1}_{j=0}, \{\varepsilon(j)\}^{n-1}_{j=1}\Big) \in \tilde{E}_{1}. 
\end{align*}
The expressions for the L.H.S and R.H.S are mentioned in Inequality \ref{Equation 7}, where $\bar{E}^c = \mathbb{R}^2 \sim \bar{E}$. Due to the noise model, a change in set of states leads to different conditional probabilities violating the equality  required for conditional independence. 
In Equation \ref{Equation 4}, $\phi^n(E)$, denotes the set of states that can be reached when the random variables $\Big(\zeta(0), Y(0), \{U(j)\}^{n-1}_{j=0}, \{\varepsilon(j)\}^{n-1}_{j=1} \Big)$ belong to $E$. Note that $\{Y(0), \{\varepsilon(j)\}^{n-1}_{j=1}\}$ impacts $\{U(j)\}^{n-1}_{j=0}$ and hence the state $\zeta(n)$ itself. That is, if we consider only the realizations of $\zeta(0)$ and $\{U(j)\}^{n-1}_{j=0}$ form the set $E$, then valid  $\zeta(n)$ could be found. However it is possible that certain realizations of $\{U(j)\}^{n-1}_{j=0}$ do not occur, i.e, has measure zero given the noise $ \{\varepsilon(j)\}^{n-1}_{j=1}$ and thus the corresponding realizations of $\zeta(n)$ do not occur. Hence, we need to include $\{Y(0), \{\varepsilon(j)\}^{n-1}_{j=1}\}$ while finding valid realizations of $\zeta(n)$. Given the control policy at stage $n$, $\Gamma_{n} : \sigma(\mathcal{F}_{n-1} \times \mathcal{B}(\{-1,1\})) \to \mathcal{P}(\mathcal{U}, \mathcal{B}(\mathcal{U}))$, the measure defined in Equation \ref{Equation 5} for can be extended to $\mathcal{F}_{n} = \sigma(\sigma(\mathcal{F}_{n-1} \times \mathcal{B}(\{-1,1\})) \times \mathcal{B}(\mathcal{U}))$. Equation \ref{Equation 6}, proves that the sequence of measures $\{\mathbb{P}_{n}\}_{n \geq 0}$ is a sequence of consistent measures on $\{(\Omega_{n}, \mathcal{F}_n)\}_{n \geq 0}$ where $\Omega_n$ is suitably defined. By \textit{Kolmogorov's Consistency Theorem}, there exists a measure $\mathbb{P}$ on $(\Omega, \mathcal{F})$ such that $\mathbb{P}(E) = \mathbb{P}_{n}(E) \forall E \in \mathcal{F}_n \forall n \in \mathbb{Z}_+$  
\subsubsection{Identification of the Classifier} \label{subsubsection 2.2.3}
Along a random path, $\{\zeta(t,\omega)\}_{t\geq 0}$, followed by the agent, consider the first $m$ data points, $\{\bar{Y}(t, \omega)\}^m_{t=1}$,  collected by the agent at the positions, $\{X(t, \omega), Z(t, \omega)\}^m_{t=1}$. For e.g. in Figure \ref{Figure 2}, the path executed by the agent results in 27 data points being collected. Due to noise, the labels at many of the data points gets flipped as indicated in the figure. Given the noisy data set, the classification problem can be formulated as follows. Let $L^{1}(\Omega_{m-1},\mathcal{F}_{m-1}, \mathbb{P}_{m-1})$ be the space of integrable random variables on the associated probability space.  The classification problem is then, 
\begin{align*}
\underset{\substack{\hspace{10pt} \\ \rho, c\in L^{1}(\Omega_{m-1},\mathcal{F}_{m-1}, \mathbb{P}_{m-1})}} \min \hspace{-0.9cm} \mathbb{E}_{\mathbb{P}_{m-1}}\Big[\sum^{m}_{j=1} Y(j)\sign(Z(j) - \rho X(j) -c)\Big].
\end{align*}
The SVM reformulation of the classification problem  is
\begin{align*}
&\underset{\substack{\hspace{10pt} \\ \rho, c\in L^{1}(\Omega_{m-1},\mathcal{F}_{m-1}, \mathbb{P}_{m-1})}} \min\;\; \mathbb{E}_{\mathbb{P}_{m-1}}\Big[\frac{\rho^2}{2}\Big]\\
\text{s.t} \;\;  \mathbb{E}_{\mathbb{P}_{m-1}} \Big[& Y(j)(Z_(j) - \rho X(j)  - c ) \Big] \geq  1, j=1, \ldots, m.
\end{align*}
\subsubsection{The Control Problem} \label{subsubsection 2.2.4}
Taking into account the identification cost, the control cost, and optimizing over control policies, the control problem is 
\begin{align*}
\underset{\{\Gamma_j\}^{m-1}_{j=0}, \rho, c \in L^1(\ldots)} \min \hspace{-0.5cm} &\mathbb{E}_{\mathbb{P}_{m-1}}\Big[\sum^{m}_{j=0} Y(j)\sign(Z(j) - \rho X(j) -c)\Big] \\
&\hspace{2.75cm} +\mathbb{E}_{\mathbb{P}_{m-1}}\Big[\sum^{m-1}_{j=0} || U(j) ||^2 \Big]\\
\text{s.t}\;\; X(t+1,\omega) &= \phi(X(t,\omega),U(t, \omega)), t=0, \ldots, m-1, \\
Y(t,\omega) =  \sign(&Z(t,\omega) - \rho^* X(t, \omega)) -c^*), t=0, \ldots, m,\\
\bar{Y}(t, \omega) &= Y(t,\omega) \varepsilon(t,\omega), t=1, \ldots, m.
\end{align*}
We note that the control policies, $\{\Gamma_j\}^{m-1}_{j=0}$, are embedded in the measure $\mathbb{P}_{m-1}$ and hence do not appear explicitly in the  cost function. 
\section{Analysis of the problem}\label{section 3}
\subsection{Deterministic Scenario}
We begin this section with some observations on the control problem formulated in subsection \ref{subsubsection 2.1.2}. Let us consider a dynamic programming approach to solve the problem. The learning cost is not evidently decomposable into a learning cost for each stage. At stage $m$, there is no control cost. Given the data points $\{x_j,z_j, y_j\}^{m}_{j=1}$, the problem is to solve the classification problem which can done using quadratic programming as mentioned before. At stage $m-1$, given the $m-1$ points visited by the agent, the objective is to chose the final data point so that classification can be performed on the resulting data set. The objective of the  classification problem is to find $(\rho,c)$ such that the minimum distance of the data points from the classifier is maximized. Suppose, the true classifier was known. If the control strategy at $m-1$ were to aid this objective, the control action would move the agent to a point where the label is opposite to the current label and whose distance is equal to the minimum distance of the given $m-1$ points from the classifier.  I

If the same argument were to be repeated at stages $j$,  $ 0 \leq j \leq m-2$, the control strategy would move the agent to points which are roughly in the neighborhood of the initial points which is not useful for the identification of the true classifier. If the learning cost is changed to $\sum^{m}_{t=1} y(t)\sign(z(t) - \rho x(t) -c)$, it also does not enforce the same. This is because, given any $m$ data points with true labels the data set is linearly separable and hence the minimum cost is always zero even though the  estimated classifier is not close to the true classifier. It appears as though the cost function for the identification problem is not utilizing the state information and the corresponding label at every stage to enhance the learning process. We explore this idea further in the following. 
\begin{figure}
\begin{center}
\includegraphics[width=\columnwidth]{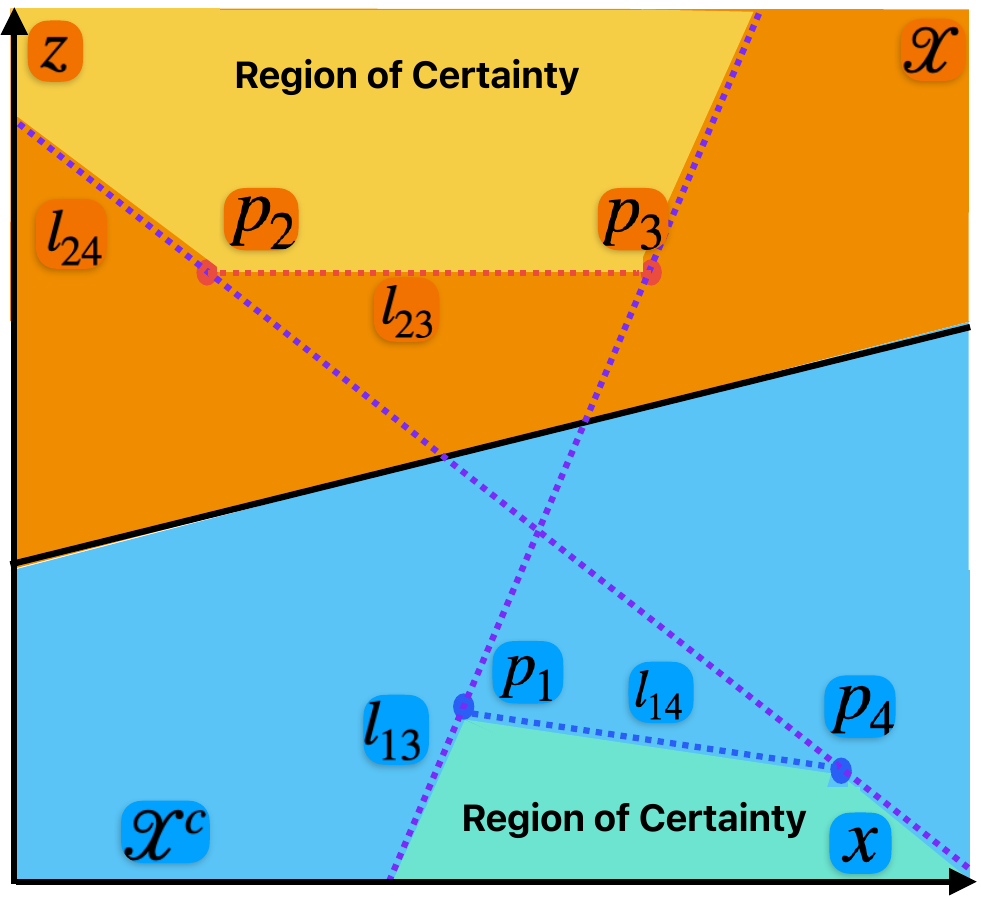}
\caption{Region of certainty from the four given points} 
\label{Figure 3}
\end{center}
\vspace{-0.9cm}
\end{figure}

We consider the scenario in depicted in Figure \ref{Figure 3} as a canonical case. The slope of the true classifier is positive and the $z$ intercept is positive as well. Other cases are : slope positive , intercept negative; slope negative intercept positive; slope negative intercept negative. The other three cases are obtained through translation and rotation of the scenario considered. Hence, the following arguments are applicable to other  cases as well. The lines $l_{13},l_{14},l_{23}, l_{24}$ are defined as $l_{ij} =\{(x,z) \in \mathbb{R}^2: z - \rho_{ij}x - c_{ij} = 0\}, \; i=1,2, j=3,4$. These lines are obtained from the four given points. Each of these lines are ``bounds" for the true classifier is the following sense. 

Consider $l_{23}$ and the region $\{(x,z) \in \mathbb{R}^2: z - \rho_{23}x - c_{23} \geq 0\}$. In this region consider any point whose $x$ co-ordinate lies between that of $p_{2}$ and $p_{3}$. The label of such a point is $1$. This is because, if there is a point with label $-1$, linear separability of data gets violated, i.e., there does no exist a linear classifier that separates $p_1, p_2, p_3,p_4$ and the point with label $-1$. However, we are unable to comment on the region $\{(x,z) \in \mathbb{R}^2: z - \rho_{23}x - c_{23} \leq 0\}$ as there is not enough information. Further, there are two points on $l_{23}$ with label $1$. Hence it is not the true classifier, but a ``bound" for the true one. The line $l_{13}$ is also a bound for the true classifier, since a line with slope slightly greater than slope of $l_{13}$ and intercept slightly less than $c_{13}$ which is in fact negative, does not separate $p_1, p_2, p_3,p_4$. By the same argument $l_{24}$ is also a bound for the true classifier. If either of them is the true classifier, then the corresponding set $\{(x,z) \in \mathbb{R}^2: z - \rho_{24}x - c_{24} \geq 0\}$ or $\{(x,z) \in \mathbb{R}^2: z - \rho_{13}x - c_{13} \leq 0\}$ carries the label $1$. Taking the intersection of the three sets, we obtain a \textit{Region of Certainty}, a set where the true label is $1$ given the four points. Using the same arguments with $l_{14}, l_{13}$, and $l_{24}$, we obtain a region of certainty with label $-1$.  

The slopes and intercepts of the four lines provide bounds on the slope and intercept of the true line. Considering the range of $\arctan$ to be $(-\frac{\pi}{2}, \frac{\pi}{2})$, let $\theta_{ij} = \arctan(\rho_{ij})$. In Figure \ref{Figure 3}, we observe that $\theta_{13} > 0, \theta_{23}\geq 0$, $\theta_{13} >  \theta_{23}$ while $\theta_{14} < 0, \theta_{24} < 0$ and $\theta_{24} < \theta_{14} $. From the linear separability arguments presented above, we conclude that the slope of the true classifier belongs to $[\theta_{24}, \theta_{14} ] \cup [\theta_{23}, \theta_{13}]$. The intercept of the true line belongs to $[c_{\min}, c_{\max}]$ where $c_{\max} = \max(c_{13}, c_{14}, c_{23},c_{24})$ and $c_{\min} = \min(c_{13}, c_{14}, c_{23},c_{24})$. In Figure \ref{Figure 3}, $c_{\max} = c_{24}$ and $c_{\min} = c_{13}$.
\begin{figure}
\begin{center}
\includegraphics[width=\columnwidth]{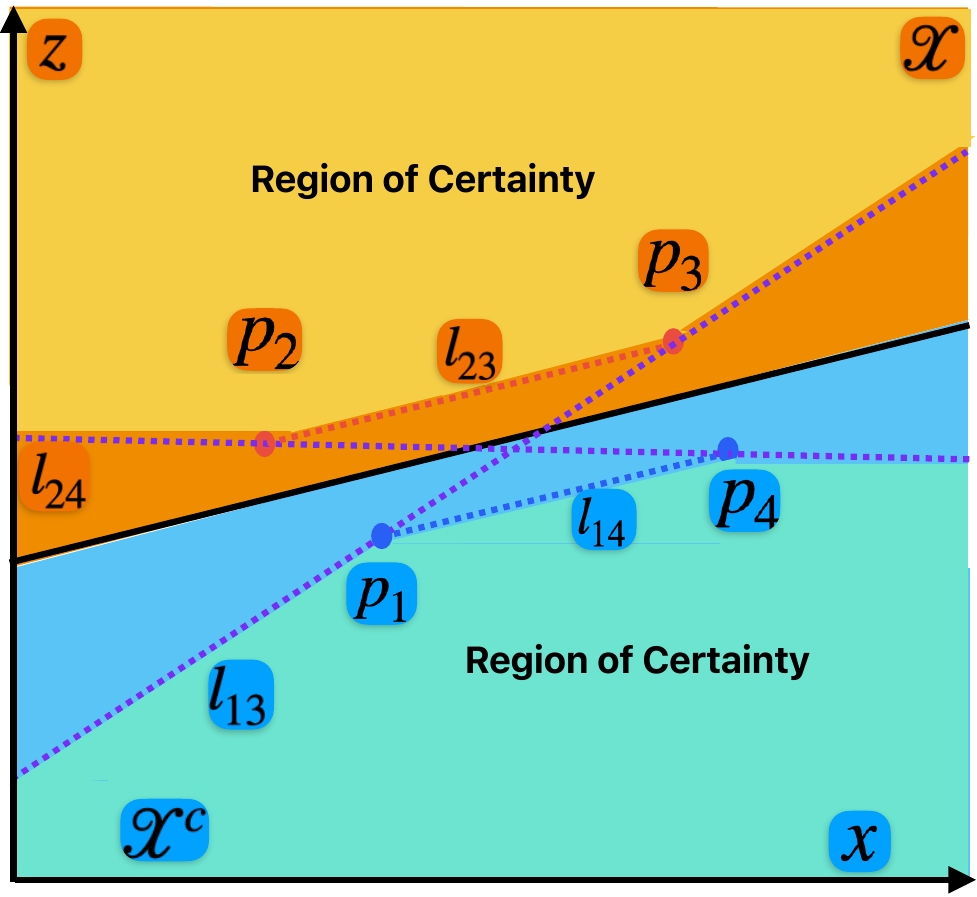}
\caption{Region of certainty from new four points obtained by Agent} 
\label{Figure 4}
\end{center}
\vspace{-0.9cm}
\end{figure}

Consider the scenario depicted in Figure \ref{Figure 4}, where the distance between pairs $p_{1}, p_{2}$ and $p_{3}, p_{4}$ which have opposite labels has reduced while the distance between  $p_{1}, p_{4}$ and $p_{2}, p_{3}$ which carry the same label has increased. In Figure \ref{Figure 4}, we observe that $\theta_{13} > 0, \theta_{23} > 0, \theta_{23} < \theta_{13} $ while $\theta_{14} > 0, \theta_{24} < 0$. The slope of the true classifier thus lies between $[\theta_{24}, 0] \cup [\min(\theta_{23} , \theta_{14}),  \theta_{13}]$ while its intercept lies between $[c_{24}, c_{13}]$. We note that the bounds for the true slope and intercept in scenario of Figure \ref{Figure 4} is a strict subset of the bounds in the scenario of Figure \ref{Figure 3} which    leads to an ``increase" in the region of certainty of both labels. 

Given the above reasoning, the control problem is to be formulated is such a away the region of certainty eventually matches with the entire regions carrying the true label. We consider one step control problems with the objective of pruning the bounds of the slope and intercept of the true classifier so that the set to which the true parameters belong eventually collapses to singletons, i.e., the true value of the parameters. To meet this objective, at a given position, in one step, the agent could either (i) move to a position of opposite label whose distance is less than the distance of the previous  point with opposite label from current position (ii)  move to a position of same label in the current region of uncertainty which is farther way from current position. 
For (i), given the current position and label, $(x(t),z(t), y(t))$, the control problem is formulated as:
\begin{align*}
\min_{u \in \mathcal{U}} || &\zeta(t+1) -\zeta(t) ||^2 + \varrho || u ||^2, \; \text{s.t} \; \; \zeta(t+1) = \phi(\zeta(t),u), \\
&  y(t) \sign(z(t+1) - \rho^*x(t+1)-c^* ) \leq 0
\end{align*}
The above problem has two issues. (i) $\rho^*$ and $c^*$ are unknown and (ii) even if the true parameters were known, the constraint may not be feasible in one step due to the limited control actions that the agent can take. Let $\mathcal{X}_{t}$ denote the region of certainty for label $1$ at stage $t$ and let $\mathcal{X}^c_{t}$ denote the region of certainty for label $-1$ at stage $t$. Let $E_{\rho,t}$ denote the bounds for the slope of the true classifier at stage $t$. We consider the following alternate formulation. 
\begin{align*}
P1 : \hspace{1cm} &\max_{u \in \mathcal{U}} || \zeta(t+1)  - \zeta(t) ||^2 - \varrho || u ||^2 \\
\text{s.t} \; \; \zeta(t+1) &= \phi(\zeta(t),u),   \zeta(t+1)  \notin \mathcal{X}^c_{t},    \zeta(t+1)  \notin \mathcal{X}_{t}  \\
& \arctan(\frac{z(t+1)  - z(t)}{x(t+1)-x(t)}) \notin E_{\rho,t}
\end{align*}
This formulation pushes the agent away to a point which is far away from its current position while not entering the regions of certainty and ensuring that the control cost is small enough. The  regions of certainty can be expressed as the intersection of a set of halfspaces. Hence, the above problem can be solved numerically. The constraint in third line is included to ensure that the agent does not end up traveling parallel to the true classifier in which case it cannot reach a point of opposite label. By moving in any direction not in set $E_{\rho,t}$, the agent is  guaranteed to reach a point of opposite label though it might take multiple steps. In implementation, it is possible to restrict the angles further, for e.g. the agent could be restricted to track the direction of the vector $\overrightarrow{p_{1}p_{2}}$ or the ``bisector" of $E_{\rho,t}$. For the agent to move to point with same label as current  label, however farther away the following problem is solved. 
\begin{align*}
P2 : \hspace{1cm} &\max_{u \in \mathcal{U}} || \zeta(t+1)  - \zeta(t) ||^2 - \varrho || u ||^2 \\
\text{s.t} \; \; \zeta(t+1) &= \phi(\zeta(t),u),   \zeta(t+1)  \notin \mathcal{X}^c_{t},    \zeta(t+1)  \notin \mathcal{X}_{t}  \\
& \arctan(\frac{z(t+1)  - z(t)}{x(t+1)-x(t)}) \in E_{\rho,t}
\end{align*}
In the above the agent could travel in a direction which is parallel to the true classifier, which is in fact good  as it meets the objective. However for all other directions, the agent could move to a point which has a opposite label to the label of its current position. 

Problems $P1$ and $P2$ are reformulations of the identification problem considered in subsection \ref{subsubsection 2.1.2}. The reformulation was necessary to incorporate feedback into the classifier identification problem. The drawback of $P1$ and $P2$ is that the control cost of the entire trajectory is not optimized but gets optimized at every stage, which may be not be optimal for entire trajectory, i.e., it is not the outcome or ``stage" optimization problem of a dynamic programming problem.  
\subsection{Stochastic Scenario}
Consider the scenario depicted in Figure \ref{Figure 2} / Figure \ref{Figure 5}. The scenario was generated by a predetermined path, i.e., no feedback involved, which was tracked by the agent. Observations were collected at $27$ positions. The labels at many of the positions are flipped as depicted in the figure. The collected data set is not linearly separable. Hence, the arguments used in the previous section are not applicable here. Given this data set there are many candidates for the true classifier. In Figure \ref{Figure 5}, we have plotted $3$ of the possible candidates. For candidate $l_{1}$ there are $6$ violations, i.e, $6$ of the observed labels are not classified correctly. For candidates $l_{2}$ and $L_{3}$ there are $7$ and $6$ violations respectively. When the last SVM algorithm was run with this data set, the resulting classifier was $l_3$ as it maximized the minimum distance from the classifier. Hence, executing a predetermined path and treating the noisy data as "deterministic" data is not helpful. 
\begin{figure}
\begin{center}
\includegraphics[width=\columnwidth]{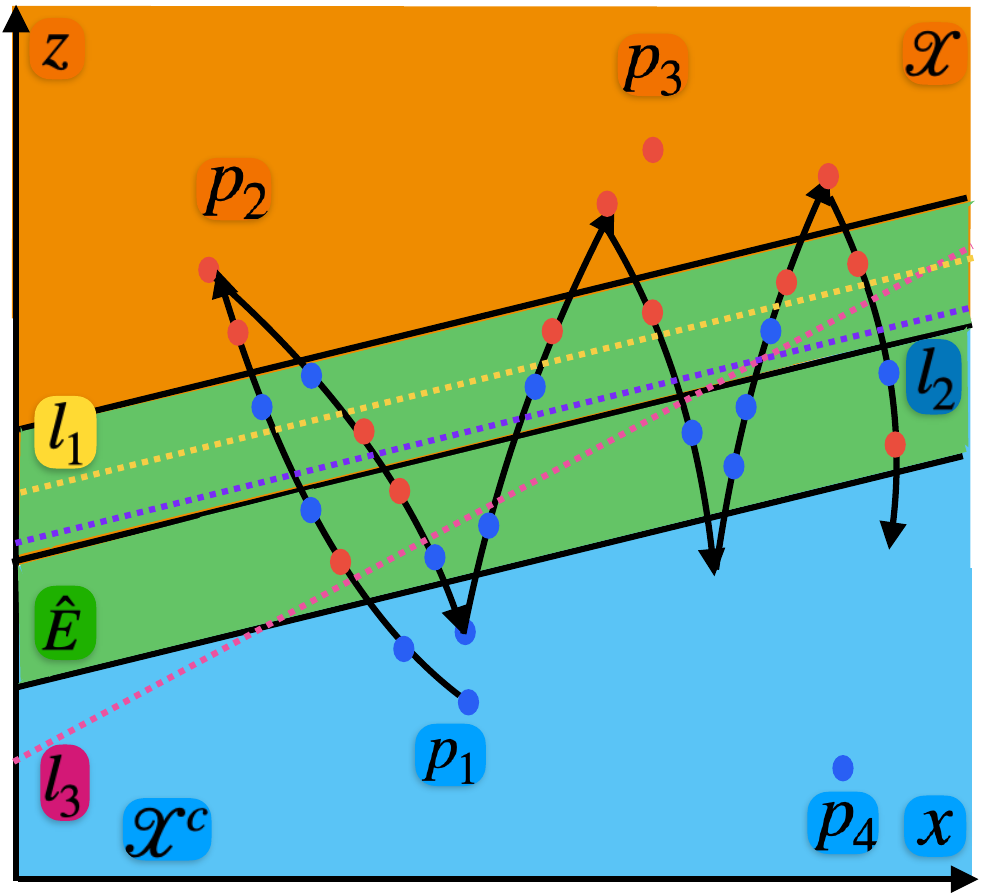}
\caption{Possible classifiers with noisy data set} 
\label{Figure 5}
\end{center}
\vspace{-0.9cm}
\end{figure}
Let $\{\hat{\mathcal{F}}_{t}\}_{t \geq 0}$ be the $\sigma$ algebra generated by $\{\{\zeta(j)\}^{t}_{j=0}, Y(0), \{\bar{Y}(j)\}^{t}_{j=1}\}$. If we consider a dynamic programming approach to solve the stochastic control problem formulated in subsection \ref{subsubsection 2.2.4}, our first observation is that the total identification cost is decomposable  in to identification cost at every stage. An "estimated" label at $(X(j), Z(j))$ could be assigned through threshold policies which are functions of $\mathbb{E}_{\mathbb{P}_{j-1}}[Y(j) | \hat{\mathcal{F}}_{j}]$. However, the computation of this conditional expectation would require the knowledge of $\mathcal{X}, \mathcal{X}^c$, i.e., the parameters of the true classifier which is not available. In the following we consider an alternate approach where the computation of this conditional expectation is not needed.

Given the four initial points with the true label known, we can find a region of certainty for either label. When the agent moves to a new point in the state space and collects an observation,i.e., corrupted label at a point in the state space, the new data point leads to an ``increase" in the region of certainity but with a certain probability. As more data points are collected the region of certainty for both labels increases but with a distribution on the region. Equivalently, the four initial points gives us a set (union of disjoint intervals), $E_{\theta,0}$ to which the true slope $\theta^* = \arctan(\rho^*)$ belongs and an interval, $I_{c,0}$, to which $c^*$ belongs. Every sequence of data points, generates a finite collection of subsets of $E_{\theta,0}/ I_{c,0}$  to which $\theta^*/ c^*$ belongs with a certain probability. Hence, the data points collected build a distribution on the set of values that the slope and the intercept can take. Our objective is to choose the data points that enhances the identification process. 

At step $t$, let $E_{\theta,t}$ denote the finite collection of subsets of $E_{\theta,0}$ generated by the data points collected until $t$ to which $\theta^*$ belongs. Each $E_{\theta,t,j} \in E_{\theta,t}$ is the disjoint union of two intervals or an interval by itself. $\theta^*$ belongs to one of the sets, $E_{\theta,t,j}$. We would like to quantify the probability with which $\theta^* \in E_{\theta,t,j}$. There are $2^t$ possibilities for $\{Y(j)\}^{t}_{j=1}$, some of which leads to a  subset $E_{\theta,t,j} \in E_{\theta,t}$. From the definition of $\bar{Y}(j, \omega)$, $Y(j, \omega) = \bar{Y}(j, \omega) \varepsilon(j,\omega)$. Given $\{\bar{Y}(j, \omega)\}^{t}_{j=1} = \{y_j\}^{t}_{j=1}$, every $\{\varepsilon(j, \omega) \}^{t}_{j=1} = \{e_{j}\}^{t}_{j=1}$ generates a possible true label sequence, $\{Y(j, \omega)\}^{t}_{j=1} = \{y_j\}^{t}_{j=1}$, as  $y_{j} =\bar{y}_j e_j , \; j=1, \ldots t$. It is not necessary that every sequence  $\{Y(j, \omega)\}^{t}_{j=1} = \{y_j\}^{t}_{j=1}$ can be used to find a subset of $E_{\theta,0}$ to which $\theta^*$ belongs. Only the sequences that lead to linearly separable data set, i.e., $\{\zeta(t,\omega), Y(t,\omega)\} \cup \{(p_{1}, -1),(p_{2}, 1),(p_{3}, 1),(p_{4}, -1)\}$ is linearly separable, are considered as valid sequences. We define $E_{t}(\cdot)$ as
\begin{align*}
&E_{t}( \{\bar{Y}(j)\}^t_{j=1}) =\{ \{\varepsilon(j) \}^{t}_{j=1} :\{\zeta(j), \bar{Y}(j)\varepsilon(j)\}^{t}_{j=1} \cup  \\
& \{\zeta(0), Y(0) \} \cup \{(p_{1}, -1),(p_{2}, 1),(p_{3}, 1),(p_{4}, -1)\}\}
\end{align*}
is linearly separable and $\bar{E}_{t}(\cdot)$ as
\begin{align*}
\bar{E}_{t}(\{Y(j)\}^t_{j=1} &| \{\bar{Y}(j)\}^t_{j=1}) = \\
 &\{ \{\varepsilon(j) \}^{t}_{j=1} :   Y(j) = \bar{Y}(j) \varepsilon(j) , \; j =1, \ldots t\}.
\end{align*}
We note that if $\{\varepsilon(j) =  e_{j} \}^{t-1}_{j=1} \notin E_{t-1}( \{\bar{Y}(j)\}^{t-1}_{j=1})$, then $\{\varepsilon(j) =  e_{j} \}^{t-1}_{j=1} \cup \{\varepsilon(t) = e_{t} \} \notin E_{t}( \{\bar{Y}(j)\}^{t}_{j=1})$ for $e_{t} =1 / -1$. Given $\{\bar{Y}(j)\}^t_{j=1}$, let $\{\varepsilon(j) \}^{t}_{j=1} \in E_{t} (\{\bar{Y}(j)\}^t_{j=1})$. For $\{Y(j) = \bar{Y}(j)\varepsilon(j)\}^{t}_{j=1}$, we define,
\begin{align*}
&\bar{\mathbb{P}}_{t}\Big( \hspace{-3pt}\{Y(j)\}^{t}_{j=1} \Big| \{\zeta(j)\}^{t}_{j=0}, \hspace{-2pt} Y(0), \{\bar{Y}(j)\}^{t}_{j=1}, E_{t}( \{\bar{Y}(j)\}^t_{j=1})\hspace{-3pt}\Big ) \\
& \overset{\Delta}{=}  \bar{\mathbb{P}}_{t} \Big(\bar{E}_{t}(\{Y(j)\}^t_{j=1} | \{\bar{Y}(j)\}^t_{j=1})  \Big| \{\zeta(j)\}^{t}_{j=0}, Y(0), E_{t}( \cdot)\Big).
\end{align*} 
The expression for this conditional probability has been stated in Equation \ref{Equation 8}. Some observations about the conditional probability are as follows. The numerator and every term in the denominator gets multiplied by $\bar{\mathbb{P}}_{t-1}(E_{z, t-1})$ but does not appear in the equation as it is a common factor. It might appear that the conditional probability does not depend on the probability of the events, $\{\zeta(j) \in F_j\}^{t-1}_{j=0} , Y(0) \in \bar{F}, \{\varepsilon(j) \in E_j\}^{t-1}_{j=1}$. However on the contrary, the probability associated with theses events is hidden  in the conditional distribution associated with the control policy, $\{\Gamma_{j}\}^{t-1}_{j=0}$.
\begin{figure*}
\hrulefill 
\begin{align}
&\bar{\mathbb{P}}_{t}\Big(\{\varepsilon(j) \in E_j\}^t_{j=1}  \Big |\Big( \{\zeta(j) \in F_j\}^t_{j=0} , Y(0) \in \bar{F} \Big),  E_{t}( \cdot) \Big) = \frac{\bar{\mathbb{P}}_{t}\Big(\{\varepsilon(j) \in E_j\}^t_{j=1}, \Big( \{\zeta(j) \in F_j\}^t_{j=1}, Y(0) \in \bar{F}  \Big) ,  E_{t}( \cdot) \Big)}{\bar{\mathbb{P}}_{t}\Big(  \Big(  \{\zeta(j) \in F_j\}^t_{j=0} , Y(0) \in \bar{F}   \Big), E_{t}( \cdot)  \Big)} \nonumber \\
&= \hspace{-4pt} \dfrac{\Big(\chi_{E_t}(1)\big( p\Gamma_{t-1}[E_{z,t-1}](U \in E_{u,1,t})  + \Gamma_{t-1}[E_{z,t-1}](U \in E_{u,2,t}) \big)  + \chi_{E_j}(-1)(1-p)\Gamma_{t-1}[E_{z,t-1}](U \in E_{u,1,t}) \Big)}{ \hspace{-20pt}\underset{\substack{\hspace{25pt} \\ \hspace{12pt} E_{t}( \{\bar{Y}(j)\}^t_{j=1}) }} \sum  \hspace{-28pt} \Big(  \hspace{-1pt} \chi_{E_t}(1)\big( p\Gamma_{t-1}[E_{z,t-1}](U \in E_{u,1,t}) \hspace{-3pt } + \Gamma_{t-1}[E_{z,t-1}](U \in E_{u,2,t}) \big)  + \chi_{E_j}(-1)(1-p)\Gamma_{t-1}[E_{z,t-1}](U \in E_{u,1,t}) \Big) } \hspace{-3pt} \label{Equation 8}\\
&E_{z,t-1} = \{ \{\zeta(j) \in F_j\}^{t-1}_{j=0} , Y(0) \in \bar{F}, \{\varepsilon(j) \in E_j\}^{t-1}_{j=1} ,  \{\varepsilon(j)\}^{t-1}_{j=1} \in  E_{t-1}( \{\bar{Y}(j)\}^{t-1}_{j=1} ) \}, \;  E_{u,1,t}(F_{t-1}) =  \{ U \in \mathcal{U}: \nonumber \\
&\phi(\zeta(t-1), U) \in F_{t} \cap  \mathbb{R}^2 \sim \bar{E}, \zeta(t-1) \in F_{t-1} \}, \;  E_{u,2,t}(F_{t-1}) = \{ U \in \mathcal{U}: \phi(\zeta(t-1), U) \in F_{t} \cap \bar{E}, \zeta(t-1) \in F_{t-1}  \} \nonumber.\\
&\mathbb{P}_{t} (\zeta(t+1) \in F_{t+1} | \mathcal{I}_{t}) = \frac{ \mathbb{P}_{t}(\zeta(t+1) \in F_{t+1} \cap \mathcal{I}_{t} )}{\mathbb{P}_{t} (\mathcal{I}_{t})}  = \frac{\sum_{\{e_{j}\}^{t}_{j=1}} \mathbb{P}_{t} (\{\varepsilon(j) = e_j\}^{t}_{j=1} \cap \zeta(t+1) \in F_{t+1} \cap \mathcal{I}_{t})}{\mathbb{P}_{t}(\mathcal{I}_{t})} \nonumber\\
&= \frac{\sum_{\{e_{j}\}^{t}_{j=1}} \Gamma_{t}[E_{\mathcal{I}_{t}}]( U \in E_{u,t} (F_t, F_{t+1}) )\mathbb{P}_{t} (E_{\mathcal{I}_{t}})}{\mathbb{P}_{t}(\mathcal{I}_{t})} = \sum_{\{e_{j}\}^{t}_{j=1}}\Gamma_{t}[E_{\mathcal{I}_{t}}]( U \in E_{u,t} (F_t, F_{t+1}) ) \bar{\mathbb{P}}_{t} ( \{\varepsilon(j) = e_j\}^t_{j=1} | \mathcal{I}_{t}),  \nonumber \\ 
& \; \text{where}, \;  E_{\mathcal{I}_{t}} = E_{z,t} =  \{ \{\varepsilon(j) = e_j\}^{t}_{j=1} \cap \mathcal{I}_{t}\} \; \text{and} \;  E_{u,t} (F_t, F_{t+1}) = \{ U \in \mathcal{U}: \phi(\zeta(t), U) \in F_{t+1}, \zeta(t) \in F_{t} \}. \label{Equation 9}
\end{align}
\hrulefill 
\vspace{-0.4cm}
\end{figure*}
We note that conditioning on the event $\{E_{t}(\{\bar{Y}(j)\}^t_{j=1})\}$ not only changes the normalization term, i.e, the denominator in Equation \ref{Equation 8} but also the event corresponding to the control policy, $\Gamma_{j}[\cdot]$.

Every valid true label sequence$, \{Y_{j}\}^{t}_{j=0}$, i.e., a sequence which is  linearly separable, generates a subset $E_{\theta, t, j}$ of $E_{\theta, 0}$  to which $\theta^*$ belongs.  We note that $E_{\theta, t, j}$ obtained from the data set need not be a subset of $E_{\theta, 0}$, however letting $E_{\theta, t, j} = E_{\theta, t, j} \cap E_{\theta, 0}$ since $\theta^* \in E_{\theta, 0}$ Similarly, the same true label sequence generates a subset $I_{c, t, j}$ of $I_{c, 0}$  to which $c^*$ belongs. The conditional probability with which $\theta^*$ and $c^*$ belongs to the sets $E_{\theta,t,j}$ and $I_{c,t,j}$ respectively is defined as the conditional probability of the corresponding true label sequence, i.e, 
\begin{align*}
\mathbb{P}_{t}(\theta^* \in E_{\theta, t, j} | \mathcal{I}_{t} ) = \mathbb{P}_{t}(c^* \hspace{-2pt} \in \hspace{-2pt} I_{c, t, j} | \mathcal{I}_{t} ) \overset{\Delta}{=} 
\bar{\mathbb{P}}_{t}\Big( \hspace{-3pt}\{Y(j)\}^{t}_{j=1} \Big| \mathcal{I}_{t}  \Big ),
\end{align*}
where $\mathcal{I}_{t} =\{\{\zeta(j)\}^{t}_{j=0}, \hspace{-2pt} Y(0), \{\bar{Y}(j)\}^{t}_{j=1}, E_{t}( \{\bar{Y}(j)\}^t_{j=1})\}$.

The objective of the control problem is to steer this distribution to the ``true" distribution of $\theta^*$ quickly. The true distribution of $\theta^*$ is that it belongs to a singleton set with probability $1$. To find this singleton set, the sets $E_{\theta,t,j}$ need to be pruned as $t$ increases and the probability with which $\theta^* \in E_{\theta,t,j}$ needs to increase.  To achieve the same we consider the following control problems. At stage $t$,
\begin{align*}
P3 : \hspace{1cm} &\max_{\Gamma_{t}} \mathbb{E}_{\mathbb{P}_{t}}\Big[|| \zeta(t+1)  - \zeta(t) ||^2 - \varrho || u ||^2  \Big| \mathcal{I}_{t} \Big] \\
\text{s.t} \; \; \zeta(t+1) &= \phi(\zeta(t),u),   \zeta(t+1)  \notin \mathcal{X}^c_{0},    \zeta(t+1)  \notin \mathcal{X}_{0}  \\
\mathbb{E}_{\mathbb{P}_{t}} \Bigg[  \arctan &(\frac{z(t+1)  - z(t)}{x(t+1)-x(t)}) | \mathcal{I}_{t}\Bigg] \notin E_{\theta,t,j}, \; \forall E_{\theta,t,j} \in E_{\theta,t}.
\end{align*}
The objective of the above control problem is move the agents in new directions, directions which are not candidates for the slope of the true line. By exploring the state space, the agent is able to get more pairs of points which are close and are of opposite label. Through this step the control problem is aiding the pruning of the sets $E_{\theta,t,j}$. Consider the following control problem:
\begin{align*}
P4 : \hspace{1cm} &\max_{\Gamma_{t}} \mathbb{E}_{\mathbb{P}_{t}}\Big[|| \zeta(t+1)  - \zeta(t) ||^2 - \varrho || u ||^2  \Big| \mathcal{I}_{t} \Big] \\
\text{s.t} \; \; \zeta(t+1) &= \phi(\zeta(t),u),   \zeta(t+1)  \notin \mathcal{X}^c_{0},    \zeta(t+1)  \notin \mathcal{X}_{0}  \\
\exists  E_{\theta,t,j} \hspace{-2pt} \in  \hspace{-2pt} E_{\theta,t} \; &\text{s.t} \; \mathbb{E}_{\mathbb{P}_{t}} \hspace{-2pt} \Bigg[  \arctan (\frac{z(t+1)  - z(t)}{x(t+1)-x(t)}) | \mathcal{I}_{t}\Bigg] \hspace{-2pt} \in E_{\theta,t,j}.
\end{align*}
The objective of the above control problem  is to ensure that agent visits regions of state space that will enhance its belief about the set to which $\theta^*$ belongs, i.e, increase or reduce the conditional probability with which $\theta^*$ belongs to a set $E_{\theta,t, j}$.  The objective is to reduce the entropy of the distribution $\bar{\mathbb{P}}_{t}\Big( \{Y(j)\}^{t}_{j=1} \Big| \mathcal{I}_{t}  \Big )$. i.e., reduce the uncertainty in the distribution. This is achieved by moving in directions which are possible candidates to the true slope. If the observed data points are as expected, then conditional probability $\theta^*$ belongs to the set of current direction is expected to increase. The solution to this problem could also lead to pruning of the sets to which $\theta^*$ belongs. From Equation \ref{Equation 9}, solving $(P3 / P4)$ invokes $\bar{\mathbb{P}}_{t}(\{\varepsilon(j) \in E_j\}^{t}_{j=1} | \mathcal{I}_{t})$, results in the control policy $\Gamma_{t}(E_{z,t})$, which is then invoked to calculate $\bar{\mathbb{P}}_{t+1}(\{\varepsilon(j) \in E_j\}^{t+1}_{j=1} | \mathcal{I}_{t+1})$. From Equations \ref{Equation 8} and \ref{Equation 9}, we conclude that at any stage $t$, by solving $P3$ and $P4$ we are manipulating $\bar{\mathbb{P}}_{t+1}(\{\varepsilon(j) \in E_j\}^{t+1}_{j=1} | \mathcal{I}_{t+1})$ through optimization of $\mathbb{P}_{t+1} (\zeta(t+1) \in F_{t+1} | \mathcal{I}_{t})$.  $P3$ and $P4$ are stochastic analogs to $P1$ and $P2$ respectively. 
\section{Algorithms}\label{section 4}
In this section, we present control algorithms to solve the problem presented in subsection \ref{subsection 1.2}. We present both deterministic and stochastic scenarios. 
\subsection{Deterministic Scenario}
The control algorithm executed by the agent is presented in Algorithm \ref{Algorithm 1}. In this algorithm, at any given position if the agent has seen a label flip from its past position through execution of $P1$ at its past position, the agent solves $P2$ and  moves to a new position. At the new position, the agent solves $P1$ and moves to new positions until it observers a label flip. The objective of $P1$ is to obtain points which are close to each other but of opposite label; hence executed multiple times until the objective is achieved. The objective of $P2$ is to obtain points which are of the same label but far from each other. It is executed only once as even if the objective is not met, it results in points with opposite labels which are far from each other. At $t=0$, the agent begins by solving $P1$ and repeats the same until a label flip is observed.

Given the data set after $m$ stages, $\{x(j), z(j), y_{j}\}^{m}_{=1}$, the distance between pairs with opposite labels is found. Of all the pairs, two pairs which have the shortest distance between them are chosen. Let the four points be $(p_{1,m}, -1)$, $(p_{2,m},1) $, $(p_{3,m}, 1)$ $(p_{4,m}, -1)$. Consider then quadrilateral formed by $(p_{1,m}, p_{2,m},p_{3,m},p_{4,m})$ and let the diagonals intersect at $p_{m}$. Consider the  angle formed by $p_{1,m}, p_{m}, p_{2,m}$. The line that bisects this angle whose slope is $\rho_m$ and intercept is $c_m$ is declared as the estimate of the true classifier.
\begin{algorithm}
\caption{Control for classification}\label{Algorithm 1}
\begin{algorithmic}[1]
\Procedure{CFC}{}
\State Given $(p_{1},-1), (p_{2}, 1), (p_{3},1), (p_{4},-1)$, $\phi(\cdot)$, and $\mathcal{U}$.
\State $\zeta(0) \gets p_{1}, y(0) \gets -1 $
\State $j \gets 0, Label \gets -1, Counter  \gets 0$
\While {$j \leq m-1$}
\If {$Counter \mod 2 =0$}
\State Solve $P1$ to obtain $U(j)$.
\State Agent moves to $\zeta(j+1)  \gets \phi(\zeta(j), U(j))$
\State $j \gets j+1$, collect observation $Y(j)$. 
\If {$Y(j) \times Label = -1$}
\State $Label \gets Y(j)$
\State $Counter \gets Counter +1$
\EndIf
\Else 
\State Solve $P2$ to obtain $U(j)$.
\State Agent moves to $\zeta(j+1)  \gets \phi(\zeta(j), U(j))$
\State $j \gets j+1$, collect observation $Y(j)$. 
\State $Counter \gets Counter +1$
\EndIf
\EndWhile
\EndProcedure
\end{algorithmic}
\end{algorithm}
\begin{proposition}
As the number of data points increases, estimated classifier converges to the true classifier, i.e., $\underset{m \to \infty} \lim \rho_{m} = \rho^*$ and $\underset{m \to \infty} \lim c_{m} = c^*$.
\end{proposition}
\begin{proof}
From the constraints of the optimization problems $(P1)$ and $(P2)$,  we note that a given point in the state space is not visited twice. This is because once the agent visits a point and the true label at that point is known, that point becomes a part of the region of certainty associated with that label. Further, a point in the region of certainty at stage $m$ is not visited by the agent at any time $n$, $n \geq m+1$. Hence, a given point is not visited twice. From earlier arguments, we note that the region of certainty for both labels is a strictly monotonic sequence of sets, $\mathcal{X}_{m} \subset \mathcal{X}_{m+1}$ and $\mathcal{X}^c_{m} \subset \mathcal{X}^c_{m+1}, \; \forall m$. Since $\mathcal{X}_{m} \subset \mathcal{X}$ and $\mathcal{X}^c_{m} \subset \mathcal{X}^c$ for all $m$, $\underset{ m \to \infty}\lim \mathcal{X}_{m} = \mathcal{X}$ and $\underset{ m \to \infty} \lim \mathcal{X}^c_{m} = \mathcal{X}^c$. The monotonicity of sequences $\{\mathcal{X}_{t}\}$ and $\{\mathcal{X}^c_{t}\}$ implies that for any $\epsilon>0$ there exists $p_{1, \epsilon}$, $p_{2, \epsilon}$ such $d(p_{1, \epsilon}, p_{2, \epsilon}) < \epsilon$ with opposite labels. Thus, there exists a sequence of pairs of points, $\{p_{1,n}, p_{2,n}\}$, such that $label(p_{1,n}) \times label(p_{2,n})  =-1 $ and the distance between the points, $d(p_{1,n}, p_{2, n})$ converges to zero. 
\begin{figure}
\begin{center}
\includegraphics[width=\columnwidth]{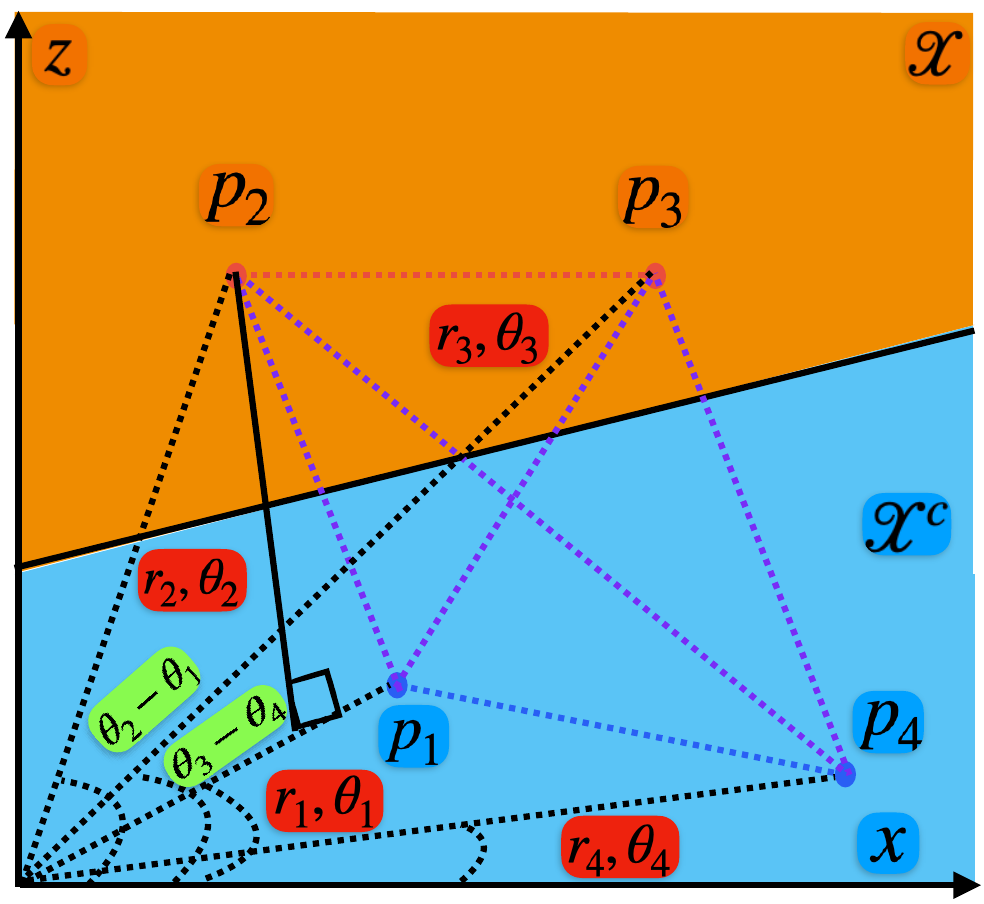}
\caption{Analysis of the problem using analytical geometry} 
\label{Figure 6}
\end{center}
\vspace{-0.9cm}
\end{figure}
Invoking the radial co-ordinates of the points as indicated in Figure \ref{Figure 6}, we note that 
\begin{align*}
&d(p_{1, \epsilon}, p_{2, \epsilon}) = \\
&\sqrt{(r_{2}\cos(\theta_2) - r_{1}\cos(\theta_1))^2 +(r_{2}\sin(\theta_2) - r_{1}\sin(\theta_1))^2}
\end{align*}
Hence as $d(p_{1, \epsilon}, p_{2, \epsilon})$ converges to $0$, $\theta_{2} - \theta_1 \to 0 $ and $r_{2} - r_{1} \to 0$, i.e., the points collapse to a point on the true classifier. Hence, the control algorithm generates points which are arbitrarily close to the true classifier. By considering 2 pairs, $p_{1, \epsilon}, p_{2, \epsilon}$ and $p_{3, \epsilon}, p_{4, \epsilon}$ such that $d(p_{1, \epsilon}, p_{2, \epsilon}) < \epsilon$ and $d(p_{3, \epsilon}, p_{4, \epsilon}) < \epsilon$  while $d(p_{1, \epsilon}, p_{4, \epsilon}) > \delta$ or $d(p_{2, \epsilon}, p_{3, \epsilon}) > \delta$, the estimated classifier obtained is then arbitrarily close to true classifier in the sense of parameter norm, $| \rho_{\epsilon} - \rho^*| < \epsilon, | c_{\epsilon} - c^*| < \epsilon$. Letting $\epsilon \to 0$, we obtain the result of the proposition. 
\end{proof}
\subsection{Stochastic Scenario}
The control algorithm  for the stochastic scenario is described in Algorithm \ref{Algorithm 2}. The agent solves $P1$ at all even time steps (including $0$) and $P2$ at all odd time steps. The most likely action from the resulting conditional distributions, $\Gamma_{j}$, is chosen at every time step.  After $m$ stages, the finite collections of sets $\{E_{\theta,m, k}\}$ and$\{I_{c,m,k}\}$ and their associated conditional probabilities $\bar{\mathbb{P}}_{t} (\{\varepsilon(k) \in E_k\}^j_{k=1} | \mathcal{I}_{j} ) $ are considered. The outcome of the algorithm are the sets $E_{\theta,m, k^*}$ and $I_{c,m, k^*}$ which have the highest conditional probability, $\bar{\mathbb{P}}_{m} (\{\varepsilon(k) = e^*_k\}^m_{k=1} | \mathcal{I}_{m} ) $, i.e, $\theta^* \in E_{\theta,m, k^*}, c^* \in I_{c,m, k^*}$ with conditional probability $\bar{\mathbb{P}}_{m} (\{\varepsilon(k) = e^*_k\}^m_{k=1} | \mathcal{I}_{m} ) $. 
\begin{algorithm}
\caption{Stochastic Control for classification}\label{Algorithm 2}
\begin{algorithmic}[1]
\Procedure{SCC}{}
\State Given $(p_{1},-1), (p_{2}, 1), (p_{3},1), (p_{4},-1)$, $\phi(\cdot)$, and $\mathcal{U}$.
\State Given $p \in [0,1]$, $\zeta(0) \gets p_{1}, y(0) \gets -1 $, $j \gets 0$
\While {$j \leq m-1$}
\If {$j \mod 2 =0$}
\State Solve $P1$ to obtain $\Gamma_j$.
\Else 
\State Solve $P2$ to obtain $\Gamma_j$.
\EndIf
\State $U(j) =\underset {u \in \mathcal{U}} \argmax \;\; \Gamma_{j}[E_{j,z}](u)$.
\State Agent moves to $\zeta(j+1)  \gets \phi(\zeta(j), U(j))$
\State $j \gets j+1$, collect observation $Y(j)$
\State Update $\{E_{\theta,j, k}\}$, $\{I_{c,j, k}\}$
\State Update $\bar{\mathbb{P}}_{j} (\{\varepsilon(k) \in E_k\}^j_{k=1} | \mathcal{I}_{j} ) $
\EndWhile
\EndProcedure
\end{algorithmic}
\end{algorithm}
Further investigation is needed to prove the convergence of the above algorithm to the true classifier. 
\section{Examples}\label{section 5}
\begin{figure}
\begin{center}
\includegraphics[width=\columnwidth]{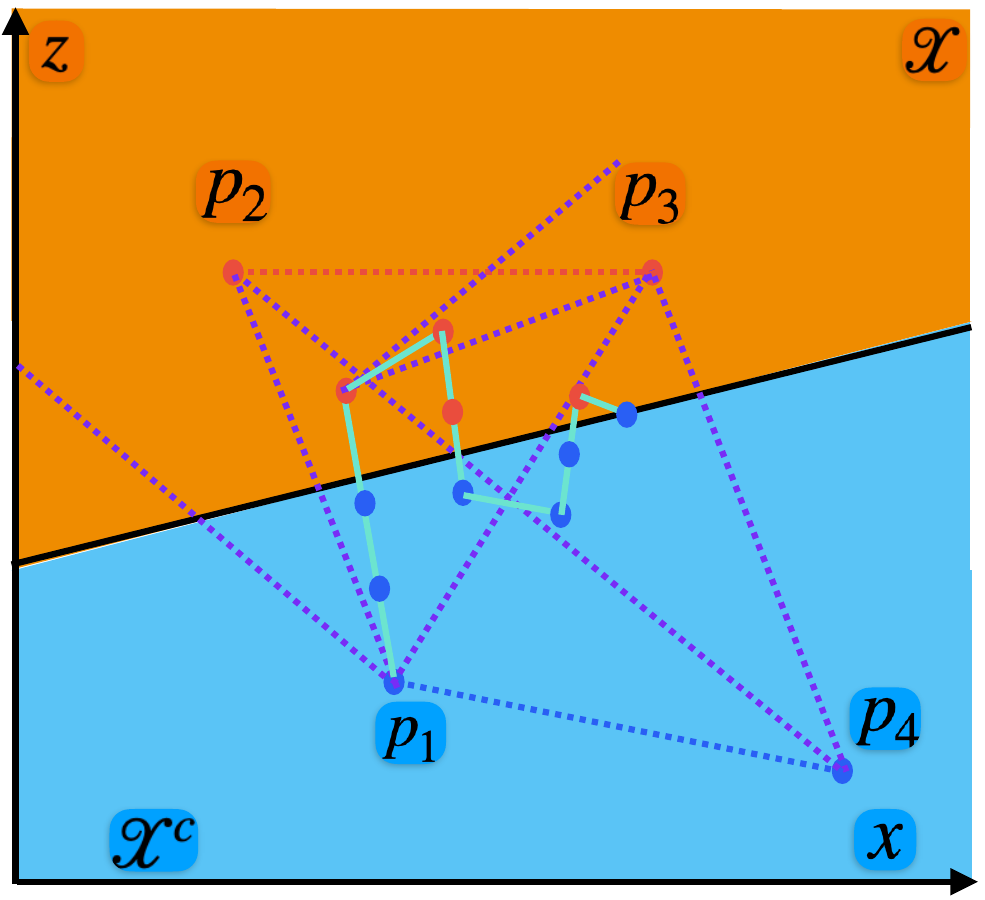}
\caption{Path of the Agent with unicycle model and noiseless observations} 
\label{Figure 7}
\end{center}
\vspace{-0.9cm}
\end{figure}
In this section, we present an example illustrating the implementation of each of the control algorithms described in the previous section. We consider a $20m \times 20m$ region in $\mathbb{R}^2$. We are given four initial points with their true labels as indicated in Figure \ref{Figure 7}.  We consider a unicycle model for the agent:
\begin{align*}
&x(t+1) = x(t) + v(t) \cos(\theta(t+1)),\\
&z(t+1) = z(t) + v(t)\sin(\theta(t+1)),\\
&\theta(t+1) = \theta(t)  + w(t), \; \text{where} \; U =[v,w]. 
\end{align*}
In the above model, each time step is a two step processes. First, $\theta$ gets updated using the angular velocity. Following this step, once the direction of travel gets fixed, the positions get updated using velocity.  With this model problem $P1$ (and similarly $P2$) gets modified to, 
\begin{align*}
P1 : \;\; &\max_{u \in \mathcal{U}} || v(t) ||^2 - \varrho ( || v(t) ||^2 + || w(t) ||^2 ) \\
&\text{s.t} \; \; \zeta(t+1) = \phi(\zeta(t),u),   \zeta(t+1)  \notin \mathcal{X}^c_{t}, \\   
&\zeta(t+1)  \notin \mathcal{X}_{t},  \; \theta (t+1) \notin E_{\rho,t}
\end{align*}
We consider $v(\cdot) \in \{0, 0.1, 0.2, \ldots, 2\}$ with unit, m/ unit time. We consider $w(\cdot) \in \{-\frac{\pi}{2}, \ldots,-0.01, 0.0, 0.01, \ldots, \frac{\pi}{2}\}$ with unit, rad/ unit time. $\varrho$ was set to $0.1$, thus enabling the agent to utilize higher values of $w$ resulting in quick change of orientation. With this setup, Algorithm \ref{Algorithm 1} was run for $10$ steps. The path followed by the agent is indicated in Figure \ref{Figure 7}. The $10$ data points gathered were utilized to estimate the classifier. The parameters corresponding to the estimated classifier were $\rho_{10} = 0.38 / \theta_{10} = 0.36 \; rad \; = 20.8 ^{\circ} \deg$ and $c_{10} = 3.6$ while corresponding values  of the true classifier where $\rho^* = 0.41 / \theta^* = 0.389 \; rad \; = 22.29 ^{\circ} \deg$ and $c_{10} = 3.5$. Thus, in this example the classifier was estimated with high accuracy. 
With the same setup, simulations were run in the stochastic setting with the probability of not flipping equal to $0.7$. The  simulation was run for $10$ steps. The trajectory followed by the agent has been plotted in Figure \ref{Figure 8}. The outcome of running Algorithm \ref{Algorithm 2} was that, it was estimated that $\theta^* \in [18.42, 25.81]^{\circ} \deg$ and $c^* \in [3.12, 3.76]$ with probability $0.84$. The estimated region of the state space to which the true classifier belongs  after $10$ steps has been marked in Figure \ref{Figure 8}. The region is bounded by lines $l_{1}$  and $l_{2}$. Further investigation is needed, to prove that this region can be shrunk to the true classifier with high probability.
\begin{figure}
\begin{center}
\includegraphics[width=\columnwidth]{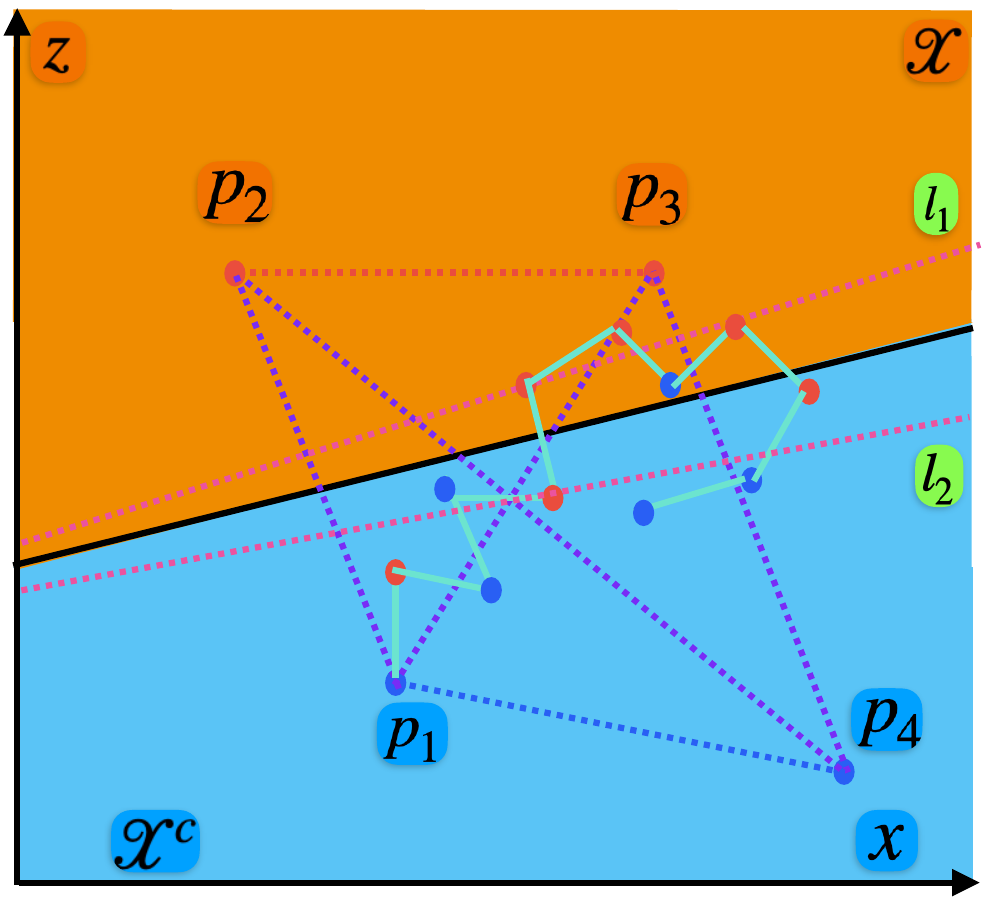}
\caption{Path of the Agent with unicycle model and noisy observations} 
\label{Figure 8}
\end{center}
\vspace{-0.9cm}
\end{figure}
\section{Conclusion and Future Work}\label{section 6}
To summarize, we considered the problem of identification of a linear classifier by an agent with noiseless and noisy data. We presented geometric interpretation of the problem which was then utilized to develop efficient control algorithms. Data obtained as a result of the control algorithms was used to identify the classifier. When the data is noiseless, we prove the convergence of the estimated classifier to the true classifier. When the data was noisy, the identification process resulted in sets to which the parameters of the true classifier belongs with high probability. 

As future work, we are interested in understanding the control problems which when analyzed using dynamic programming would result in value functions which are obtained through optimization problems similar to the one step control problems formulated in this paper. We would like to investigate the convergence of the estimated classifier to the true classifier in the stochastic case in suitable topology. Formal connections to dual control and adaptive sampling are to be established.
\bibliographystyle{IEEEtran}
\bibliography{biblio}
\end{document}